\documentclass[letterpaper, 10 pt, conference]{ieeeconf}
\IEEEoverridecommandlockouts
\usepackage{verbatim}
\usepackage{cite}
\usepackage{times}
\usepackage{amsmath}
\usepackage{mathrsfs}
\usepackage{amsbsy}
\usepackage{amsxtra}
\usepackage{amsopn}
\usepackage{latexsym}
\usepackage{amsfonts}
\usepackage{epsfig}
\usepackage{graphicx}
\usepackage{amssymb}
\usepackage{stmaryrd}
\usepackage{mathabx}
\usepackage{epstopdf}

\newtheorem{lemma}{Lemma}
\newtheorem{remark}{Remark}
\newtheorem{definition}{Definition}
\newtheorem{theorem}{Theorem}

\begin{document}
\pagenumbering{gobble}

\graphicspath{{./images/}}

\thispagestyle{plain} \pagestyle{plain} \date{}
\title{Binary Log-Linear Learning with Stochastic Communication Links}
\author{Arjun~Muralidharan,~Yuan~Yan~and~Yasamin~Mostofi
\thanks{This work is supported in part by NSF NeTS award $\#$ 1321171.}
\thanks{The authors are with the Department of Electrical and Computer Engineering,
University of California Santa Barbara, Santa Barbara, CA 93106, USA email:
$\{$arjunm, yuanyan, ymostofi$\}$@ece.ucsb.edu.}}
\maketitle

\vspace{-0.5in}

\begin{abstract} 

In this paper, we consider distributed decision-making over stochastic communication links in multi-agent systems. We show how to extend the current literature on potential games with binary log-linear learning (which mainly focuses on ideal communication links) to consider the impact of stochastic communication channels. More specifically, we derive conditions on the probability of link connectivity to achieve a target probability for the set of potential maximizers (in the stationary distribution). Furthermore, our toy example demonstrates a transition phenomenon for achieving any target probability for the set of potential maximizers.

\end{abstract}

\section{Introduction}\label{sec:introduction}

Non-cooperative game theory has recently emerged as a powerful tool for the distributed control of multi-agent systems \cite{marden_cooperative_potential,martinez_cov_game,shamma_cov_game}. By designing proper local utility functions and learning algorithms that satisfy certain properties, desirable global behaviors can be achieved.
Potential games \cite{monderer_potential} are an important class of non-cooperative games and have recently received considerable attention in the literature \cite{marden_welfare}. In potential games, the local utility function of the agents is aligned with a potential function in order to achieve a global objective through local decisions. 

There are a number of learning algorithms that can guarantee the convergence to a Nash equilibrium for potential games such as fictitious play \cite{monderer_fictitious} and joint strategy fictitious play\cite{marden_JSFP}. However, a Nash equilibrium may be a sub-optimum outcome and not the potential maximizer. Log-linear learning (first introduced in \cite{blume_log-linear}), on the other hand, is a learning mechanism that can guarantee convergence to the set of potential maximizers. As a result, it has been the subject of considerable research recently \cite{marden_revisiting}. Binary log-linear learning \cite{marden_revisiting},\cite{arslan_autonomous} is a variant of log-linear learning which can further handle constrained actions sets, i.e.\ scenarios where the future actions of the players are limited based on their current action (like in robotic networks). 

While considerable progress has been made for distributed decision making using potential games, ideal communication links are often assumed. In other words, it is typically assumed that an agent can hear from all the other agents that will impact its utility function. In realistic communication environments with packet-dropping stochastic communication links, this is simply not possible. For instance, Fig.\ \ref{fig:Allscales_blue_gray} shows an example of real channel measurements.  We can see that the channel exhibits a great degree of stochasticity due to the shadowing and multipath fading components. Thus, it is the goal of this paper to bring an understanding of the impact of stochastic packet-dropping communication links on potential games with binary log-linear learning, where each link is properly represented with an action-dependent probability of connectivity.  By extending \cite{marden_revisiting}, we derive conditions on the temperature (defined in Section \ref{sec:BLLL}) and probabilities of connectivity to achieve a given target probability (in the stationary distribution) for the set of potential maximizers (Theorems \ref{thm:blll-comm} and \ref{thm:blll-comm-2}). In Section \ref{sec:toy_game}, in a toy example, we further observe a transition behavior for achieving any target probability.

\section{Problem Setup}\label{sec:problemsetup}

In this section, we first introduce some basic concepts and properties of potential games. We then review the binary log-linear learning algorithm and the theory of resistance trees, which we use in our subsequent analysis. Finally, we motivate the need for considering stochastic communication links.
\subsection{Potential Game (see \cite{drew_game_theory} for more details)}\label{sec:background}
A game $\mathcal{G} = \{\mathcal{I},\{\mathcal{A}_{i}\}_{i\in \mathcal{I}},\{U_i\}_{i\in \mathcal{I}}\}$ is defined by its three components:
\begin{enumerate}
\item $\mathcal{I}=\{1,2,\cdots,n\}$ is the set of players/agents/robots;
\item $\mathcal{A}_{i}$ is the set of all the actions (choices) that agent $i$ has. Then, an action profile $a=(a_1,\cdots,a_n)\in\mathcal{A}$ denotes the collection of actions of all the agents, where $\mathcal{A}=\mathcal{A}_{1}\times \cdots \times \mathcal{A}_{n}$ is the space of all action profiles;
\item $U_{i}:\mathcal{A} \to \mathbb{R}$ is the utility function of agent $i$.
\end{enumerate}
One of the most important concepts in game theory is that of a pure Nash equilibrium, which is defined as follows.

\begin{definition}[Pure Nash Equilibrium]\label{def:nasheq}
Consider a game $\mathcal{G} = \{\mathcal{I},\{\mathcal{A}_{i}\}_{i\in \mathcal{I}},\{U_i\}_{i\in \mathcal{I}}\}$. An action profile $a^{\text{NE}}$ is said to be a pure Nash equilibrium of the game if and only if
\begin{equation}\label{eq:nasheq}
U_{i}(a^{\text{NE}}) \geq U_{i}(a_{i},a_{-i}^{\text{NE}}), \quad \forall \;a_{i} \in  \mathcal{A}_{i} \text{ and } \forall\; i \in \mathcal{I},
\end{equation}
where $a_{-i}=(a_1,\cdots,a_{i-1},a_{i+1},\cdots,a_n) \in \mathcal{A}_{-i} $ denotes the action profile of all the agents except $i$ and $\mathcal{A}_{-i} = \mathcal{A}_{1}\times \cdots \mathcal{A}_{i-1} \times \mathcal{A}_{i+1} \times \cdots \times \mathcal{A}_{n}$.
\end{definition}

As can be seen, a game has reached a pure Nash equilibrium if and only if no agent has the motivation to unilaterally change its action. In this paper, we are interested in potential games.  Potential games can have broad applications in distributed multi-robot systems since they allow each robot to make local decisions while a global objective (the potential function) is optimized.

\begin{definition}[Potential Games \cite{monderer_potential}]\label{def:potential_game}
A game $\mathcal{G} = \{\mathcal{I},$ $\{\mathcal{A}_{i}\}_{i\in \mathcal{I}},\{U_i\}_{i\in \mathcal{I}}\}$ is said to be a potential game with potential function $\phi:\mathcal{A} \to \mathbb{R}$ if
\begin{align}\label{eq:potential_game}
&U_{i}(a'_{i},a_{-i}) -U_{i}(a_{i},a_{-i})  = \phi(a'_{i},a_{-i}) -\phi(a_{i},a_{-i}),\nonumber\\
&\hspace{0.5in} \forall\; a_{i},a'_{i} \in  \mathcal{A}_{i} \text{,}\; \forall \;a_{-i} \in  \mathcal{A}_{-i} \text{ and } \forall\; i \in \mathcal{I}.
\end{align}
\end{definition}

As can be seen, a potential game requires a perfect alignment between the potential function and the agents' local utility functions. It is straightforward to confirm that the action profile that maximizes the potential function is a pure Nash equilibrium. Hence, a pure Nash equilibrium is guaranteed to exist in potential games. 

\subsection{Binary Log-Linear Learning (see \cite{marden_revisiting})}\label{sec:BLLL}

In several scenarios, the set of possible actions that an agent can take is limited by its current action. For instance, in multi-robot systems, the next possible position of an agent is limited by its current position and velocity. Formally, we refer to this limited set as an agent's constrained action set i.e., $\mathcal{A}_{i}^{\text{cons}}(a_{i}) \subseteq \mathcal{A}_{i}$ is agent $i$'s constrained action set where $a_{i}$ is its current action.

Binary Log-linear learning (BLLL) is a variant of Log-linear learning (as shown in \cite{marden_revisiting}) which can handle constrained action sets. It is summarized as follows. At each time step $t$, an agent $i \in \mathcal{I}$ is chosen randomly (uniformly) and is allowed to alter its action.\footnote{Note that the selection does not require coordination among the nodes and can be achieved through each agent using a Poisson clock \cite{boyd_randomized}.} All the other agents repeat their previous actions, i.e. $a_{-i}(t)=a_{-i}(t-1)$. Agent $i$ then plays according to the following strategy:
\begin{align}
p_{i}^{a_{i}(t-1)}(t) & = {e^{{1 \over \tau}U_{i}(a(t-1))} \over {e^{{1 \over \tau}U_{i}(a(t-1))}+e^{{1 \over \tau}U_{i}(\hat{a}_{i},a_{-i}(t-1))}}}, \label{eq:blll-1}
\\ p_{i}^{\hat{a}_{i}}(t) & = {e^{{1 \over \tau}U_{i}(\hat{a}_{i},a_{-i}(t-1))} \over {e^{{1 \over \tau}U_{i}(a(t-1))}+e^{{1 \over \tau}U_{i}(\hat{a}_{i},a_{-i}(t-1))}}},\label{eq:blll-2}
\end{align}
where $\hat{a}_{i}$ is an action that is chosen uniformly from the constrained action set $\mathcal{A}^{\text{cons}}_{i}(a_{i}(t-1))$, $p_{i}^{a_{i}(t-1)}(t)$ is the probability of repeating its previous action, $p_{i}^{\hat{a}_{i}}(t)$ is the probability of selecting action $\hat{a}_{i}$, and $\tau > 0$ is the temperature.

Moreover, the constrained action sets should possess the following two properties:
\begin{definition}[Reachability]\label{def:reachability}
For all $i \in \mathcal{I}$ and any action pair $a_{i}^{0},a_{i}^{m} \in \mathcal{A}_{i}$, there exists a sequence of actions $a_{i}^{0} \to a_{i}^{1} \to \cdots \to a_{i}^{m-1} \to a_{i}^{m}$ satisfying $a^{k}_{i} \in \mathcal{A}^{\text{cons}}_{i}(a^{k-1}_{i})$, $\forall \; k \in \{1,\cdots,m\}$.
\end{definition}
\begin{definition}[Reversibility]\label{def:reversability}
For all $i \in \mathcal{I}$ and any action pair $a_{i}^{0},a_{i}^{1} \in \mathcal{A}_{i}$, if $a^{1}_{i} \in \mathcal{A}^{\text{cons}}_{i}(a^{0}_{i})$, then we have $a^{0}_{i} \in \mathcal{A}^{\text{cons}}_{i}(a^{1}_{i})$.
\end{definition}

Note that Definition \ref{def:reachability} implies that any action profile in $\mathcal{A}$ can be reached in finite time steps. Definition \ref{def:reversability} means that each agent can go back to its previous action.

\begin{theorem} (see \cite{marden_revisiting})\label{thm:blll}
Consider a potential game with constrained action sets that satisfy the reachability and reversibility properties. BLLL ensures that the support of the stationary distribution is the set of potential maximizers, as $\tau\rightarrow 0$.
\end{theorem}

We next introduce the concept of an asynchronous best reply process over constrained action sets, which is a process where each agent locally improves its own utility function when it is its turn to alter its action. 

\begin{definition}\label{def:best-reply}
An asynchronous best reply process over constrained action sets is defined as follows. At each time $t>0$, an agent $i$ is randomly chosen (uniformly) and allowed to alter its action. All other agents repeat their current action, i.e. $a_{-i}(t) = a_{-i}(t-1)$. Agent $i$ then selects an action $\hat{a}_{i}$ uniformly from its constrained action set, i.e. $\hat{a}_{i}\sim\text{unif}(\mathcal{A}_{i}^{\text{cons}}(a_{i}(t-1)))$. It then plays the action which maximizes its utility function: ${a}_{i}(t) \in \big\{a_{i}\in\left\{\hat{a}_i,a_{i}(t-1)\right\}:U_{i}(a_{i},a_{-i}(t-1)) = $ $\max\left\{ U_{i}(a(t-1)),U_{i}(\hat{a}_{i},a_{-i}(t-1))\right\}\big\}$.
\end{definition}

The best reply process does not necessarily maximize the overall potential function of the game as it may result in a suboptimal Nash equilibrium. When $\tau=0$, the BLLL algorithm boils down to an asynchronous best reply process on the constrained action sets. A $\tau>0$ then allows each agent to occasionally select locally suboptimal moves, i.e. it selects an action that decreases its local utility with a non-zero probability. These occasional suboptimal moves are useful as they prevent the agents from converging to a suboptimal Nash equilibrium. 
The BLLL algorithm can then be thought of as a perturbation of the asynchronous best reply process, where the size of the perturbation is controlled by the temperature $\tau$. This idea is formalized in Section \ref{sec:BLLL_perturbed}.

\subsection{Resistance Trees}\label{sec:resistance}
In this part, we briefly review the concept of resistance trees, which we will use in our subsequent sections. We refer the readers to \cite{young_evolution} for a detailed discussion.

\subsubsection{Resistance Trees}\label{sec:resistance_trees}
Let $P^{0}$ be a stationary Markov chain defined on a state space $X$. We call this the \textit{unperturbed} process. The process $P^{\epsilon}$ is then called a \textit{regular perturbed Markov process} if it satisfies the following conditions:
\begin{enumerate}
\item $P^{\epsilon}$ is aperiodic and irreducible;
\item $\lim_{\epsilon \rightarrow 0} P^{\epsilon}(x \rightarrow y) = P^{0}(x \rightarrow y)$, $\forall \;x,y \in X$, where $P^{\epsilon}(x \rightarrow y)$ and $P^{0}(x \rightarrow y)$ denote the transition probabilities from state $x$ to $y$ of processes $P^{\epsilon}$ and $P^{0}$ respectively;
\item if $P^{\epsilon}(x \rightarrow y)>0$, for some $\epsilon >0$, then there exists some $R(x \rightarrow y) \geq 0$, such that $0<{\lim_{\epsilon \rightarrow 0}} \epsilon^{-R(x \rightarrow y)}P^{\epsilon}(x \rightarrow y) < \infty$, where we refer to $R(x \rightarrow y)$ as the resistance of the transition from state $x$ to $y$.
\end{enumerate}

Construct a tree $T$ with $\left\vert{X}\right\vert$ vertices, one for each state, rooted at some vertex $z$, such that there exists a unique directed path to $z$ from every other vertex. The weight of a directed edge from vertex $x$ to $y$ is given by the resistance $R(x \rightarrow y)$. Such a tree is called a resistance tree whose resistance is given by the sum of the $\left\vert{X}\right\vert - 1$ edges that compose it. Since $P^{\epsilon}$ is aperiodic and irreducible (the first condition of the regular perturbed Markov process), there exists a unique stationary distribution $\mu^{\epsilon}$ for a given $\epsilon$. Define $p^{\epsilon}_z = \sum_{T\in\mathcal{T}_z}\prod_{[x,y]\in T}P^{\epsilon}(x \rightarrow y)$, where $[x,y]$ is the directed edge from vertex $x$ to $y$ and $\mathcal{T}_{z}$ denotes the set of all the trees that are rooted at $z$. Then, we have
\begin{align}\label{eq:res-2}
\mu^{\epsilon}_z = \frac{p^{\epsilon}_z}{\sum_{z'\in X} {p^{\epsilon}_{z'}}},
\end{align}
where $\mu_{z}^{\epsilon}$ denotes the probability of state $z$ in the stationary distribution.

The stochastic potential of state $z$ is then defined as the minimum resistance among all the trees that are rooted at $z$:
\begin{equation}\label{eq:res-1}
\gamma(z) = \min_{T \in \mathcal{T}_{z}} \sum_{[x,y] \in T}R(x \rightarrow y),
\end{equation}

\begin{theorem} (see \cite{young_evolution})\label{thm:resistance-tree}
Let $P^{\epsilon}$ be a regular perturbed Markov process of $P^{0}$ and $\mu^{\epsilon}$ be its unique stationary distribution. Then
\begin{enumerate}
\item $\lim_{\epsilon \rightarrow 0}\mu^{\epsilon} = \mu^{0}$ exists,\footnote{The perturbations effectively select one of the stationary distributions of $P^0$.} where $\mu^{0}$ is a stationary distribution of $P^{0}$;
\item We have $\mu_{x}^{0}>0$ iff $\gamma(x) \leq \gamma(y)$, $\forall \; y \in X$, where $\mu_{x}^{0}$ denotes the probability of state $x$ in the stationary distribution $\mu^0$.
\end{enumerate}
\end{theorem}

Theorem \ref{thm:resistance-tree} shows that the stochastically stable states (the support of the stationary distribution $\mu^{0}$) are the states with the minimum stochastic potential, i.e. $\mu_{x}^{0}>0$ if and only if $x$ minimizes $\gamma(x)$. 

Informally, resistance (of a transition) is a measure of how difficult that transition is. The greater the resistance, the more difficult (less likely) the transition. So the resistance of a tree rooted at state $x$ is a measure of how difficult it is for other states to transit to $x$. Thus, a state with minimum stochastic potential is a state to which it is the easiest to get to (informally speaking) as compared to other states. We will utilize this metaphor of difficulty in Section \ref{sec:toy_game} when explaining some of our results.

\subsubsection{BLLL as a Regular Perturbed Markov process}\label{sec:BLLL_perturbed}
BLLL algorithm induces a regular perturbed Markov process with the unperturbed process corresponding to the asynchronous best reply process defined in Section \ref{sec:BLLL} \cite{marden_revisiting}. The probability of a feasible transition $a^{0} \rightarrow a^{1} = (a_{i}^{1},a^{0}_{-i})$ (where agent $i$ alters its action and $a^0,a^1 \in \mathcal{A}$) is then given by
\begin{equation}\label{eq:log-res-3}
P^{\epsilon}(a^{0} \rightarrow a^{1}) = {1 \over n|\mathcal{A}_{i}^{\text{cons}}(a_{i}^0)|}{\epsilon^{-U_{i}(a_{i}^{1},a_{-i}^{0})} \over { \epsilon^{-U_{i}({a}_{i}^{1},a_{-i}^{0})} + \epsilon^{-U_{i}({a}_{i}^{0},a_{-i}^{0})}}},
\end{equation}
where $ \epsilon = e^{-\frac{1}{\tau}}$. As shown in \cite{marden_revisiting}, the resistance of this transition is as follows:
\begin{equation}\label{eq:log-res-4}
R(a^{0} \rightarrow a^{1}) = V_{i}(a^{0},a^{1}) - U_{i}(a^{1}),
\end{equation}
where $V_{i}(a^{0},a^{1}) = \max \{ U_{i}(a^{0}),U_{i}(a^{1})\}$.

Based on the theory of resistance trees, it can be shown that only the action profiles that maximize the potential function have the minimum stochastic potential \cite{young_evolution}. This in turn means that the stochastically stable states of BLLL are the set of potential maximizers, as also stated in Theorem \ref{thm:blll}.

\subsection{Stochastic Communication Links}\label{sec:comm_links}

Most of the current research in the area of motion planning of multi-robot systems assumes over-simplified channel models. For instance, it is common to assume perfect links or links that are perfect within a certain radius. In reality, however, communication links are best modelled stochastically. More specifically, the received channel to noise ratio (CNR) is a multi-scale random process with three major components: distance-dependent path loss, shadowing and multipath fading \cite{goldsmith_wireless}. See Fig. \ref{fig:Allscales_blue_gray} for a real example.

In the current literature on potential games, it is assumed that each agent is connected to all the other agents that can impact its next step utility function for all the possible actions in its constrained set. If the wireless channel is modeled as a disk with a known radius, then this can be achieved by properly designing the constrained action set. However, in the case of realistic communication links, this is simply not the case.  More specifically, it is not possible for every agent to truly evaluate its utility function as other agents with whom it cannot communicate may be influencing it. Thus, realistic communication links have a considerable implication for distributed decision making using potential games. It is the goal of this paper to bring an understanding of their impact on BLLL and derive sufficient conditions (on link quality and temperature) to guarantee a target probability for the set of potential maximizers (in the stationary distribution) in the presence of stochastic links.

\begin{figure}[h]
\centering
\mbox{\epsfig{figure=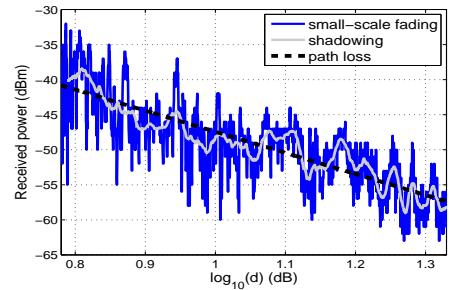,height=1.6in,width=2.6in}}
\vspace{-0.15in}
\caption{Underlying dynamics of the received signal power across an indoor route \cite{malmirchegini_spatial}.}
\vspace{-0.1in}
\label{fig:Allscales_blue_gray}
\end{figure}

\section{Impact of Stochastic Communication Links on Binary Log-Linear Learning}\label{sec:impact}

In this section, we characterize the impact of imperfect communication on the outcome of the BLLL algorithm. We first prove that given any arbitrarily-high probability of the set of potential maximizers, there exist a connectivity probability and temperature $\tau$ that can achieve it. We then give an illustrative example to provide a deeper understanding of our results.
\subsection{BLLL with Stochastic Communication Links}\label{sec:BLLL_links}
Consider the case where the communication graph among the agents is given by an undirected random graph $\mathcal{C}(a) = (\mathcal{I}, \mathcal{E}(a))$, where  $\mathcal{E}(a)$ denotes the set of edges, i.e.\ the communication links among the agents. Then, the probability of having a link (probability of connectivity) between agents $i$ and $j$ is given by $p_{\text{c},j,i}(a) = p_{\text{c},i,j}(a) = p_{\text{c}}(a_i, a_j)$, where we take $p_{\text{c},i,i}(a) = 1$, for all $a\in \mathcal{A}$ and $i\in\mathcal{I}$.\footnote{Notation not to be confused with $p_{z}^{\epsilon}$, which was used in (\ref{eq:res-2}).} Note that we have taken the probability of connectivity (and subsequently the communication graph) to be action dependent to make our analysis more general (which naturally implies a time-varying graph). For instance, when the action of an agent involves its position, then the probability of connectivity becomes a function of the action profile. We assume that the probability of connectivity of different links are independent of each other in this paper. We further assume that the communication graph is drawn independently in each iteration.

As mentioned in Section \ref{sec:BLLL}, in each iteration of the BLLL algorithm, an agent is chosen randomly (uniformly) to alter its action. Meanwhile, a realization of the communication graph is drawn from the random graph $\mathcal{C}(a)$. Let $\mathcal{I}_{\text{c},i}$ be the corresponding realization of the set of agents that agent $i$ can communicate with. The probability of realization $\mathcal{I}_{\text{c},i}$ is given by $p_{\text{c},i}(\mathcal{I}_{\text{c},i},a) = \prod_{j\notin\mathcal{I}_{\text{c},i}}(1-p_{\text{c},i,j}(a)) \prod_{j\in\mathcal{I}_{\text{c},i}} p_{\text{c},i,j}(a)$. Note that $\mathcal{I}_{\text{c},i} = \mathcal{I}$ corresponds to the case where agent $i$ can hear from all the other agents. Also, since the probability of connectivity is state-dependent, $p_{\text{c},i}(\mathcal{I}_{\text{c},i},a)$ is also a function of $a$.

The agent then has to assess its local utility and determine its action based on incomplete information. To represent this, we extend the definition of the utility function $U_{i}:\mathcal{A} \to \mathbb{R}$ such that it is well defined for all $U_{i}(a|\mathcal{I}_{\text{c},i})$, $\forall a \in \mathcal{A},\forall \mathcal{I}_{\text{c},i}$, where $U_{i}(a|\mathcal{I}_{\text{c},i})$ is the evaluated local utility function of agent $i$ given that it only communicates with agents in $\mathcal{I}_{\text{c},i}$. One possibility for evaluating $U_{i}(a|\mathcal{I}_{\text{c},i})$ is that the agent ignores the impact of agents not in $\mathcal{I}_{\text{c},i}$. Another possible strategy is for an agent to assume the last communicated action for the agents it is unable to communicate with.\footnote{However, evaluating which is a better strategy becomes case dependent and is an avenue for future work.}

In order to evaluate the impact of the stochastic communication links on the learning dynamics, we start with a temperature-dependent probability of connectivity of the form $p_{\text{c},i,j}(a) = {1\over 1+\epsilon^{m_{i,j}(a)}},\;\forall i,j\in \mathcal{I},\;\forall a\in \mathcal{A}$, where $m_{i,j}(a)>0$ is a constant. Based on our assumed form, we always have $p_{\text{c},i,j}(a) > 0.5$. Note that for $p_{\text{c},i,j}(a) = p_{\text{c}},\;\forall a\in \mathcal{A},\; \forall i,j \in \mathcal{I}$, $p_{\text{c}}$  need not have this temperature-dependent form, as we will show in our result (Theorem \ref{thm:blll-comm-2}). The probability of the transition $a^0 \rightarrow a^1 = (a^1_i, a^0_{-i})$ in the presence of stochastic communication links can then be characterized as follows:
\begin{align} 
 & P^{\epsilon}_c(a^{0} \rightarrow a^{1})  = \sum_{\mathcal{I}_{\text{c},i}}{p_{\text{c},i}(\mathcal{I}_{\text{c},i},a^0)P^{\epsilon}_c(a^{0} \rightarrow a^{1}|\mathcal{I}_{\text{c},i})} \nonumber \\
& = {1 \over n|\mathcal{A}_{i}^{\text{cons}}(a_{i}^0)|}\sum_{\mathcal{I}_{\text{c},i}}{p_{\text{c},i}(\mathcal{I}_{\text{c},i},a^0)\epsilon^{-U_{i}(a^{1}|{\mathcal{I}_{\text{c},i}})} \over { \epsilon^{-U_{i}({a}^{1}|{\mathcal{I}_{\text{c},i}})} + \epsilon^{-U_{i}({a}^{0}|{\mathcal{I}_{\text{c},i}})}}} \nonumber
\end{align}
\scriptsize
\vspace{-0.1in}
\begin{align*}
 =  {1 \over n|\mathcal{A}_{i}^{\text{cons}}(a_{i}^0)|}\sum_{\mathcal{I}_{\text{c},i}}\frac{\epsilon^{-U_{i}(a^{1}|{\mathcal{I}_{\text{c},i}}) + \sum_{j\notin \mathcal{I}_{\text{c},i}} m_{i,j}(a^0)}}{(\epsilon^{-U_{i}({a}^{1}|{\mathcal{I}_{\text{c},i}})} + \epsilon^{-U_{i}({a}^{0}|{\mathcal{I}_{\text{c},i}})})\prod_{j \in \mathcal{I}}(1+\epsilon^{m_{i,j}(a^0)})}.
\end{align*}
\normalsize
It can be seen that expressing the probability of connectivity in this fashion ensures that BLLL in the presence of stochastic communication links induces a regular perturbed Markov process with the unperturbed process as the asynchronous best reply process (Definition \ref{def:best-reply}).

\begin{figure*}[!htb]
\begin{equation}\label{eq:blll-alg-comm-2}
P^{\epsilon}_c(a^{0} \rightarrow a^{1})  = {1 \over n|\mathcal{A}_{i}^{\text{cons}}(a_{i}^0)|}\sum_{\mathcal{I}_{\text{c},i}}\frac{\epsilon^{V_{i}({a}^{0}, {a}^{1}|{\mathcal{I}_{\text{c},i}}) - U_i({a}^{1}|{\mathcal{I}_{\text{c},i}}) + \sum_{j\notin \mathcal{I}_{\text{c},i}} m_{i,j}(a^0)}}{\left(\epsilon^{V_{i}({a}^{0}, {a}^{1}|{\mathcal{I}_{\text{c},i}})-U_i({a}^{1}|{\mathcal{I}_{\text{c},i}})} + \epsilon^{V_{i}({a}^{0}, {a}^{1}|{\mathcal{I}_{\text{c},i}}) -U_i({a}^{0}|{\mathcal{I}_{\text{c},i}})}\right) \prod_{j \in \mathcal{I}}\left(1+\epsilon^{m_{i,j}(a^0)}\right)}.
\end{equation}
\hrulefill
\end{figure*}

We can further show that the equation above can be expressed as shown in (\ref{eq:blll-alg-comm-2}) on top of the next page, which results in the following expression for the resistance of this transition:
\begin{align}
& R_{c}(a^{0} \rightarrow a^{1}) = \min_{\mathcal{I}_{\text{c},i}} \left\{ R_{c}(a^{0} \rightarrow a^{1}| \mathcal{I}_{\text{c},i}) + \sum_{j\notin \mathcal{I}_{\text{c},i}} m_{i,j}(a^0)\right\},\label{eq:res_comm}
\end{align}
where $R_{c}(a^{0} \rightarrow a^{1}| \mathcal{I}_{\text{c},i}) = V_{i}(a^{0},a^{1}|{\mathcal{I}_{\text{c},i}})  - U_{i}({a}^{1}|{\mathcal{I}_{\text{c},i}})$ and $V_{i}({a}^{0}, {a}^{1}|{\mathcal{I}_{\text{c},i}}) = \max\{U_i({a}^{0}|{\mathcal{I}_{\text{c},i}}),$ $U_i({a}^{1}|{\mathcal{I}_{\text{c},i}})\}$. Note that $R_{c}(a^0\to a^1|\mathcal{I}) = R(a^0\to a^1)$, where $R(a^0\to a^1)$ is the resistance in case of perfect communication (see (\ref{eq:log-res-4})).

It can be seen that imperfect communication affects the transition probability from $a^0 \rightarrow a^1$, and as a result, affects its resistance. This means that the stochastically stable states may change. Hence, the outcome of the game may be significantly different as compared to the case of perfect communication.

\begin{lemma} \label{lem:suff-bound}
Consider a potential game where the agents employ BLLL algorithm in the presence of stochastic communication links. Furthermore, consider constrained action sets that satisfy the reachability and reversibility properties. The states with the minimum stochastic potential are the set of potential maximizers if we have the following,
\begin{align}
\sum_{j\notin\mathcal{I}_{\text{c},i}} m_{i,j}(a^0)\geq
R(a^{0} \rightarrow a^{1}) -R_{c}(a^{0} \rightarrow a^{1}| \mathcal{I}_{\text{c},i}),\label{eq:suff-bound-1}
\end{align}
for every agent $i\in\mathcal{I}$, all $\mathcal{I}_{\text{c},i}$ and all $a^{0} \rightarrow a^{1} = (a_{i}^1,a_{-i}^0)$, where $R(a^{0} \rightarrow a^{1})$ is the resistance for the case of perfect communication (see (\ref{eq:log-res-4})).
\end{lemma}

\begin{proof}
If the conditions in the lemma hold, then the resistance of the transition from  $a^{0}$ to $a^{1} = (a^1_i, a^0_{-i})$, for some agent $i$, becomes $R_{c}(a^{0} \rightarrow a^{1}) = R(a^{0} \rightarrow a^{1})$. Therefore, the resistances of the transitions do not change as compared to the case of perfect communication. The proof of the lemma then follows immediately from Lemma 5.2 and Theorem 5.1 in \cite{marden_revisiting}.
\end{proof}

\begin{remark}\label{remark:suff-bound}
A good choice of the constants $\{m_{i,j}(a)\}_{i,j\in \mathcal{I},\;a \in \mathcal{A}}$, is such that they satisfy $m_{i,j}(a) \geq \max_{a^{0}\rightarrow a^{1}}\{R(a^{0} \rightarrow a^{1})\}$. This has the advantage that there are separate conditions for each $m_{i,j}(a)$ and that they are not dependent on how communication failures affect the game.
\end{remark}

\begin{remark}
Lemma \ref{lem:suff-bound} provides sufficient conditions to guarantee that the states with the minimum stochastic potential are still the set of potential maximizers. Equation (\ref{eq:suff-bound-1}) can be more explicitly expressed as a function of connectivity as follows:
\begin{align}
\sum_{j\notin\mathcal{I}_{\text{c},i}} m_{i,j}(a^0) & = \sum_{j\notin\mathcal{I}_{\text{c},i}} \log_{\epsilon}(\epsilon^{m_{i,j}(a^0)})\nonumber\\
& = \sum_{j\notin\mathcal{I}_{\text{c},i}} \log_{\epsilon}{1-p_{\text{c},i,j}(a^0) \over p_{\text{c},i,j}(a^0)}\nonumber\\
& \geq R(a^{0} \rightarrow a^{1}) -R_{c}(a^{0} \rightarrow a^{1}| \mathcal{I}_{\text{c},i}),\nonumber
\end{align}
for all $\mathcal{I}_{\text{c},i}$ and all $a^{0} \rightarrow a^{1}$. We can see that $\log_{\epsilon}{1-p_{\text{c},i,j}(a) \over p_{\text{c},i,j}(a)}$ is an important parameter (always positive). Furthermore, if $R_{c}(a^{0} \rightarrow a^{1}| \mathcal{I}_{\text{c},i}) \geq R(a^{0} \rightarrow a^{1}) $, then connectivity  to agent $i$ is not important, since the condition is always satisfied.

The following theorem shows that we can find some $\tau>0$ and $p_{\text{c},i,j}(a)<1$ to guarantee that the probability of the set of potential maximizers in the stationary distribution is larger than or equal to some required threshold.
\end{remark}

\begin{theorem}\label{thm:blll-comm}
Consider a potential game where the agents employ BLLL algorithm in the presence of stochastic communication links. Furthermore, consider constrained action sets that satisfy the reachability and reversibility properties. For any given $p_{\text{tar}} < 1$, there exists a $\tau_{\text{th}} > 0$, such that the probability of the set of potential maximizers in the stationary distribution is larger than or equal to $p_{\text{tar}}$, if $0<\tau \leq \tau_{\text{th}}$ and $p_{\text{c},i,j}(a) =  {1\over 1+\epsilon^{m_{i,j}(a)}},\; \forall a \in \mathcal{A},\; \forall i,j \in \mathcal{I}$, where $\{m_{i,j}(a)\}_{i,j\in \mathcal{I},\;a \in \mathcal{A}}$ are constants satisfying Lemma \ref{lem:suff-bound}.
\end{theorem}
\begin{proof}
We construct temperature-dependent probabilities of connectivity, as discussed in Section \ref{sec:BLLL_links}, such that the constants $\{m_{i,j}(a)\}_{i,j\in \mathcal{I},\;a \in \mathcal{A}}$ satisfy Lemma \ref{lem:suff-bound}. Then, we have a regular perturbed Markov process, and the states with the minimum stochastic potential are still the set of potential maximizers. 

From Theorem \ref{thm:resistance-tree}, we have $\mu_x^0=\lim_{\tau \to 0}\mu_{x}^{\epsilon} = 0$, $\forall x\notin \mathcal{A}^{*}$, where $\mathcal{A}^{*} \subseteq \mathcal{A}$ is the set of potential maximizers. Thus, we know that, for any $ x \notin \mathcal{A}^{*}$, there exists a $\tau_{x}>0$ such that $\mu_{x}^{\epsilon}\leq{1-p_{\text{tar}}\over|\mathcal{A}\setminus\mathcal{A}^{*}|} $, if $0 < \tau\leq \tau_{x}$. Hence, we have
\begin{align*}
\sum_{a\in \mathcal{A}^{*}}\mu_{a}^{\epsilon} & = 1 - \sum_{a\notin \mathcal{A}^{*}}\mu_{a}^{\epsilon} \geq p_{\text{tar}},
\end{align*}
if $0< \tau \leq \tau_{\text{th}} = \min_{x\notin \mathcal{A}^{*}}\tau_{x}$.
\end{proof}

The following theorem shows that we can find a sufficient lower bound on the probability of connectivity, to ensure that the probability of the set of potential maximizers is larger than or equal to some required threshold, for the special case of probabilities of connectivity which are state independent and equal for all links, i.e., $p_{\text{c},i,j}(a) = p_{\text{c}}, \; \forall a \in \mathcal{A}, \; \forall i,j \in \mathcal{I}$.

\begin{theorem}\label{thm:blll-comm-2}
Consider a potential game where the agents employ BLLL algorithm in the presence of stochastic communication links with probabilities of connectivity that are state independent and equal for all links, i.e., $p_{\text{c},i,j}(a) = p_{\text{c}}, \; \forall a \in \mathcal{A}, \; \forall i,j \in \mathcal{I}$. Furthermore, consider constrained action sets that satisfy the reachability and reversibility properties. For any given $p_{\text{tar}} < 1$, there exists a $p_{\text{c},\text{th}} < 1$, such that the probability of the set of potential maximizers in the stationary distribution is larger than or equal to $p_{\text{tar}}$, if $p_{\text{c}} \geq p_{\text{c},\text{th}}$, where $\tau = {- m\over \ln\left({1-p_{\text{c}}\over p_{\text{c}}}\right)}$, with $m$ representing a constant that satisfies Lemma \ref{lem:suff-bound}.
\end{theorem}

\begin{proof}
We construct a temperature-dependent probability of connectivity, as discussed in Section \ref{sec:BLLL_links}, such that the constant $m$ satisfies Lemma \ref{lem:suff-bound}, in order to establish a $p_{\text{c},\text{th}}$.
From Theorem \ref{thm:blll-comm}, we know that there exists a $\tau_{\text{th}}>0$, such that probability of convergence to the set of potential maximizers is larger than or equal to $p_{\text{tar}}$, if $0<\tau\leq\tau_{\text{th}}$.  

Consider $p_{\text{c}} \geq p_{\text{c},\text{th}} = {1\over 1+e^{-m\over\tau_{\text{th}}}}$. Then, the temperature associated with this probability of connectivity is $\tau = {- m\over \ln\left({1-p_{\text{c}}\over p_{\text{c}}}\right)} = \tau_{\text{th}}{\ln\left({1-p_{\text{c},\text{th}}\over p_{\text{c},\text{th}}}\right)\over\ln\left({1-p_{\text{c}}\over p_{\text{c}}}\right)}\leq\tau_{\text{th}}$. Thus, the probability of the set of potential maximizers is larger than or equal to $p_{\text{tar}}$, if $p_{\text{c}} \geq p_{\text{c},\text{th}} $ with $\tau = {- m\over \ln\left({1-p_{\text{c}}\over p_{\text{c}}}\right)}$.
\end{proof}

Recall that the BLLL algorithm allows perturbations from the asynchronous best reply process, so as to allow agents to intentionally choose locally sub-optimum actions with a small probability. However, in the imperfect communication case, the failure of communication links also causes the agents to unintentionally make sub-optimum decisions. Given some $p_{\text{tar}}$, our given sufficient conditions then aim to restrict how often the unintentional mistakes can be made. One observation in the proof of Theorem \ref{thm:blll-comm} is that in general, $\tau_{\text{th}}$ and $p_{\text{c},i,j}(a)$ are decreasing and increasing functions of $p_{\text{tar}}$ respectively, i.e.\ if higher $p_{\text{tar}}$ is required, then smaller temperature and better connectivity are needed. This is because when $p_{\text{tar}}$ is higher, then less perturbation from the asynchronous best reply process is allowed. Hence, the agents have to assess their local utilities more accurately, which requires better connectivity.

Theorem \ref{thm:blll-comm} can be easily extended to the case of Log-Linear Learning. We skip the details for brevity.
\subsection{An illustrative example}\label{sec:toy_game}
In this part, we provide an example to have a better understanding on the impact of imperfect communication and the results of Theorem \ref{thm:blll-comm}. Consider a 2-agent game where the action sets of the players are given by $ \mathcal{A}_{1} = \{\text{T},\text{B}\} $ (top, bottom) and $ \mathcal{A}_{2} = \{\text{L},\text{R}\}$ (left, right). The utility function $U_{i}:\mathcal{A}=\mathcal{A}_{1} \times \mathcal{A}_{2} \to \mathbb{R}$  is given by
\begin{center}
\begin{tabular}{ r|c|c| }
\multicolumn{1}{l}{$U_{1}$}
 &  \multicolumn{1}{c}{$\text{L}$}
 & \multicolumn{1}{c}{$\text{R}$} \\
\cline{2-3}
$\text{T}$ & 1 & 3 \\
\cline{2-3}
$\text{B}$ & 3 & 1 \\
\cline{2-3}
\end{tabular}     \hspace{0.4in}                 \begin{tabular}{ r|c|c| }
\multicolumn{1}{l}{$U_{2}$}
 &  \multicolumn{1}{c}{$\text{L}$}
 & \multicolumn{1}{c}{$\text{R}$} \\
\cline{2-3}
$\text{T}$ & 1 & 2 \\
\cline{2-3}
$\text{B}$ & 4 & 1 \\
\cline{2-3}
\end{tabular}
\end{center}
This is a potential game with the following potential function $\phi:\mathcal{A} \rightarrow \mathbb{R}$:
\begin{center}
\begin{tabular}{ r|c|c|c| }
\multicolumn{1}{l}{}
 &  \multicolumn{1}{c}{$U_{1}$}
 & \multicolumn{1}{c}{$U_{2}$}
& \multicolumn{1}{c}{$\phi$}\\
\cline{2-4}
$a^1 = (\text{B},\text{R})$ & 1 & 1 & 1 \\
\cline{2-4}
$a^2 = (\text{T},\text{L})$ & 1 & 1 & 2\\
\cline{2-4}
$a^3 = (\text{T},\text{R})$& 3 & 2 & 3\\
\cline{2-4}
$a^4 = (\text{B},\text{L})$ & 3 & 4 & 4\\
\cline{2-4}
\end{tabular}
\end{center}

For the case where the two nodes cannot communicate, we take $U_{1}(a|\{1\}) = \left\{\begin{array}{ll}3 & \text{if }a_1 = \text{T}\\ 1 & \text{if }a_1 = \text{B}\end{array}\right.$ and $U_{2}(a|\{2\}) = \left\{\begin{array}{ll}1 & \text{if }a_2 = \text{L}\\ 2 & \text{if }a_2 = \text{R}\end{array}\right.$.


We next find the stochastic potential for each action profile, i.e.\ the minimum resistance of the tree rooted at each action profile, by using (\ref{eq:res-1}). For the case of perfect communication, the (only) state with minimum stochastic potential is $a^4$. For the case of imperfect communication, we consider the scenario where the probability of connectivity is state-independent, i.e., $p_{\text{c}}=p_{\text{c},1,2} = p_{\text{c},2,1}$. Let $m = m_{1,2}=m_{2,1}=\log_{\epsilon}{1-p_{\text{c}} \over p_{\text{c}}}$. By constructing all the resistance trees and minimizing (\ref{eq:res-1}), we find that $a^4$ still has the minimum stochastic potential if and only if $m = \log_{\epsilon}{1-p_{\text{c}} \over p_{\text{c}}}>1$. From (\ref{eq:suff-bound-1}) in Lemma \ref{lem:suff-bound}, our derived sufficient condition for the probability of connectivity to guarantee that $a^4$ still has the minimum stochastic potential can be found as $m = \log_{\epsilon}{1-p_{\text{c}} \over p_{\text{c}}}\geq 3$.

\begin{figure}[h]
\centering
\mbox{\epsfig{figure=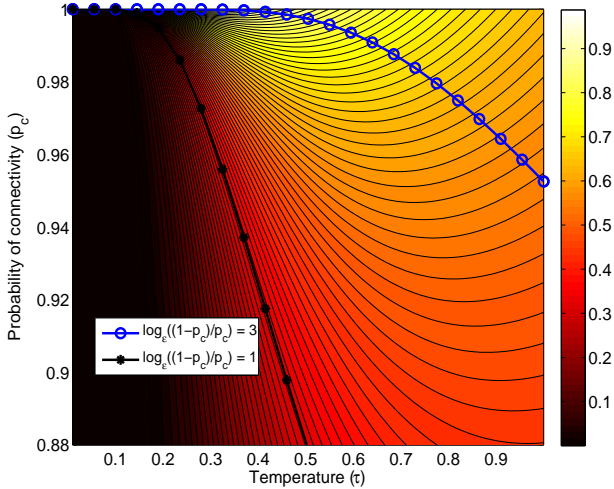,height=2.8in,width=3.6in}}
\vspace{-0.3in}
\caption{The color map corresponds to the probability of the potential maximizer as a function of $\tau$ and $p_{\text{c}}$. The blue line with empty circle markers shows $p_{\text{c}} = {1 \over 1+e^{-3/\tau}}$, while the black line with filled circle markers shows $p_{\text{c}} =  {1 \over 1+e^{-1/\tau}}$.}
\vspace{-0.1in}
\label{fig:var_temp_m}
\end{figure}

By evaluating (\ref{eq:res-2}), we calculate the stationary distribution for the case of imperfect communication to see how the probability of the potential maximizer ($a^4$) changes as a function of $p_{\text{c}}$ and $\tau$, as shown in Fig.\ \ref{fig:var_temp_m}. The blue line with empty circle markers in the figure represents the curve $p_{\text{c}} = {1 \over 1+e^{-3/\tau}}$ ($m=3$), while the black line with filled circle markers represents $p_{\text{c}} =  {1 \over 1+e^{-1/\tau}}$ ($m=1$). It can be seen that given any $p_{\text{tar}}$, the required probability of the potential maximizer can always be achieved by choosing a fixed $m>1$ and finding some appropriate $\tau$ and $p_{\text{c}}$. (Theorem \ref{thm:blll-comm} shows a sufficient condition for this, where $m\geq 3$). Informally, as discussed in Section \ref{sec:resistance}, this is because the state with the minimum stochastic potential (in this case, the potential maximizer $a^4$) is the easiest to transit to. The curve $m=1$, i.e., $p_{\text{c}} =  {1 \over 1+e^{-1/\tau}}$ could be thought of as a transition curve, as curves above it can achieve any $p_{\text{tar}}$, while those below it cannot. This is due to the fact that above this curve $a^4$ is the only state with minimum stochastic potential, while, on this curve, the states with minimum stochastic potential are $a^4$ and $a^3$. On the other hand, below this curve, $a^3$ is the (only) state with minimum stochastic potential. Informally, this means that the potential maximizer, $a^4$, becomes more difficult to transit to as compared to $a^3$ when $m < 1$. In fact, for $m<1$, the probability of the potential maximizer $a^4$ becomes arbitrarily small as $\tau \to 0$. Note that we have plotted the y-axis only down to 0.88 for better visibility. Also, note that a case of connectivity of 96\%, for instance, means that the packets are dropped 4\% of the time, which is a typical value for several scenarios\cite{boyce1998packet}.

\begin{figure}[h]
\centering
\mbox{\epsfig{figure=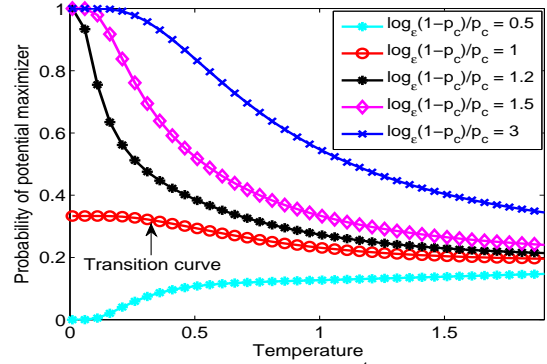,height=2in,width=3.2in}}
\vspace{-0.15in}
\caption{Probability of potential maximizer ($a^4$) as a function of temperature for different values of $\log_{\epsilon}{1-p_{\text{c}} \over p_{\text{c}}}$.}
\vspace{-0.15in}
\label{fig:m_curves}
\end{figure}

Fig. \ref{fig:m_curves} and \ref{fig:var_m} better highlight the transition behavior. Fig.\ \ref{fig:m_curves} shows the probability of the potential maximizer as a function of the temperature, for various values of $\log_{\epsilon}{1-p_{\text{c}} \over p_{\text{c}}}$ ($m$). The transition behavior of the curve $m = \log_{\epsilon}{1-p_{\text{c}} \over p_{\text{c}}} = 1$, is clearly observed. Finally, Fig.\ \ref{fig:var_m} shows the probability of the potential maximizer as a function of $\log_{\epsilon}{1-p_{\text{c}} \over p_{\text{c}}}$ for various values of temperature $\tau$. The transition point can clearly be seen at $\log_{\epsilon}{1-p_{\text{c}} \over p_{\text{c}}} = 1$. 

\begin{figure}[h]
\centering
\mbox{\epsfig{figure=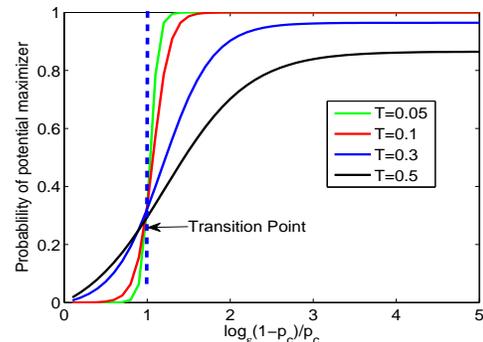,height=1.9in,width=2.8in}}
\vspace{-0.15in}
\caption{Probability of potential maximizer ($a^4$) vs $\log_{\epsilon}{1-p_{\text{c}} \over p_{\text{c}}}$ for different temperatures.}
\vspace{-0.27in}
\label{fig:var_m}
\end{figure}

\section{Conclusions}
In this paper, we considered the problem of distributed decision-making in multi-agent systems (via potential games) with an emphasis on the impact of realistic communication links. We showed how to extend the current literature on potential games with binary log-linear learning to account for stochastic communication channels. We derived conditions on the probabilities of link connectivity and BLLL's temperature to achieve a target probability for the set of potential maximizers. Furthermore, our toy example demonstrated a transition phenomenon for achieving any target probability.

\bibliography{ref}
\bibliographystyle{unsrt}

\end{document}